\documentclass[12pt]{article} %To run in WinEdt 6.0 press "Shift+Ctl+P"

\usepackage{amsmath,amssymb}
\usepackage[english]{babel}
\usepackage{bbm,amsthm}

\usepackage{titlesec}
\titleformat*{\section}{\large \bfseries}
\titlespacing{\section}{0pt}{\parskip}{\parskip}

\titleformat*{\subsection}{\normalsize \bfseries}
\titlespacing{\subsection}{0pt}{\parskip}{\parskip}

\usepackage{setspace}

\numberwithin{equation}{section}

\newtheorem{theorem}{Theorem}[section]
\newtheorem{lemma}[theorem]{Lemma}
\newtheorem{corollary}[theorem]{Corollary}
\newtheorem{proposition}[theorem]{Proposition}

\title{\Large{Asian option as a fixed-point}}%An iterative pricing formula}}
\author{%Saul D. Jacka\thanks{Department of Statistics, The University of Warwick (s.d.jacka@warwick.ac.uk).} \and
Adriana Ocejo\thanks{Department of Mathematics and Statistics, UNC Charlotte, amonge2@uncc.edu, +1(704)687-1413.
This work was supported, in part, by funds provided by the University of North Carolina at Charlotte.}}

\date{April, 2018}

\begin{document}

\maketitle

\begin{abstract}
We characterize the price of an Asian option, a financial contract, as a fixed-point of a non-linear operator.
In recent years, there has been interest in incorporating changes of regime into the parameters describing the evolution of the underlying asset price,
namely the interest rate and the volatility, to model sudden exogenous events in the economy.
Asian options are particularly interesting because the payoff depends on the integrated asset price.
We study the case of both floating- and fixed-strike Asian call options with arithmetic averaging
when the asset follows a regime-switching geometric Brownian motion with coefficients that depend on a Markov chain.
The typical approach to finding the value of a financial option is to solve an associated system of coupled partial differential equations.
Alternatively, we propose an iterative procedure that converges to the value of this contract with geometric rate using a classical fixed-point theorem.
\end{abstract}


%60J70-Applications of Brownian motions and diffusion theory
%60K37-Processes in random environments
%91G20-Derivative securities
%47H10-Fixed-point theorems
%91B25-Asset pricing models

\section{Introduction}

In this paper we use a fixed-point theorem to characterize the price of a floating-strike Asian call option (defined below) when the interest rate and volatility of the underlying asset are subject to changes of regime during the pricing period.
We next formulate the problem precisely.
%This is an extension of BLA and it
%is motivated by the empirical results of BLA.

Let $(\Omega,\mathcal{F},P)$ be a probability space which supports a Brownian motion $B=(B_t)_{t\geq 0}$
and a continuous-time Markov chain $Y=(Y_t)_{t\geq 0}$ independent of $B$ with finite state space $\mathcal{M}=\{1,2,\ldots,m\}$ and generator $Q=(q_{i j})_{m\times m}$,
\[
q_{i j}\geq 0 \quad \mbox{for $i\neq j$}, \qquad \sum_{j\in \mathcal{M}} q_{i j}=0, \qquad q_i:=-q_{i i}\geq 0.
\]

Suppose that under $P$, the underlying asset price follows the regime-switching geometric Brownian motion
\[
dX_t=X_t[ (r(Y_t)-\delta)\,dt+\sigma(Y_t)dB_t\,], \qquad 0\leq t\leq T,
\]
where $r(i)>0$ and $\sigma(i)>0$ denote the risk-free interest rate and the volatility at regime $i$, respectively,
and $\delta\geq 0$ is the dividend rate. Denote by $\mathcal{F}_t$ the sigma-algebra generated by $\{(X_u,Y_u):0\leq u\leq t\}$.

Throughout this paper we fix a time $t_0\in [0,T)$, and define the integrated process
\[
A_t:=\int_{t_0}^t X_u\,du, \qquad t_0\leq t\leq T
\]
and for convenience, we extend the definition of $A$ to $A_t=0$ for $t\in [0,t_0]$.

The European call option has payoff $(X_T-K)^+$ at time $T$, where $K>0$ is a fixed strike.
An Asian option is a path-dependent European-style option, where
the payoff depends on the average of past prices during the time interval $[t_0,T]$.
%Typically, the initiation time of the averaging period is $t_0=0$.
Asian options are mainly classified as {\it fixed-strike} (when $X_T$ is replaced by $A_T$ and the strike $K$ is fixed) or {\it floating-strike} (when $K$ is replaced by $A_T$).
In this paper we study both cases.

More precisely, the price at time $s\in[0,T]$ of an Asian call option with floating-strike expiring at $T$ is given by
\begin{equation} \label{eq:Asianprice}
C(s,x,a,i)=\mathbb{E}_{s,x,a,i}\left[ e^{-\int_{s}^T r(Y_u)du} \left(X_T-\frac{A_T}{T-t_0} \right)^+\right],
\end{equation}
while an Asian call option with fixed-strike expiring at $T$ with strike price $K$ is
\begin{equation} \label{eq:Asianprice_fixed}
C_K(s,x,a,i)=\mathbb{E}_{s,x,a,i}\left[ e^{-\int_{s}^T r(Y_u)du} \left(\frac{A_T}{T-t_0} -K\right)^+\right],
\end{equation}
where we use the notation $\mathbb{E}_{s,x,a,i}[\cdot]$ for $\mathbb{E}[\,\cdot\mid X_s=x,A_s=a,Y_s=i]$, $x>0, a\geq 0$.
%Similarly, the price at time $s$ of an Asian put option with floating-strike expiring at $T$ is given by
%\begin{equation} \label{eq:AsianPutprice}
%P(s,x,a,i)=\mathbb{E}_{s,x,a,i}\left[ e^{-\int_{s}^T r(Y_u)du} \left(\frac{A_T}{T-t_0}-X_T\right)^+\right], \qquad t_0\leq s\leq T.
%\end{equation}
The options are referred to as {\it starting} when $s=t_0$, {\it in-progress} when $s>t_0$, and {\it forward-starting} when $s<t_0$.

Regime-switching processes in finance were initially proposed by Hamilton in his economic studies with discrete time models
on the effect of incorporating shifts in the parameters of the model via an unobserved discrete time two-state Markov chain (see \cite{Ham88},\cite{Ham89}).
%term structure of interest rates and the analysis of business cycle
Since then, several pricing methods for financial instruments have emerged under the assumption of regime-switching coefficients.
Such models successfully incorporate sudden changes in the economy and compensate some of the drawbacks of the classical Black-Scholes model due to the constancy of the drift and volatility parameters.
To mention some literature, Buffington and Elliott \cite{BE2002b}, Yao et al. \cite{YZZ2006}, and Zhu et al. \cite{Zhuetal2012}
concentrate on vanilla European options;
Guo and Zhang \cite{GZ2004} study perpetual American put options; and
Chan and Zhu \cite{CZ2015} deal with barrier options.

%Also, Boyaerchenko and Levendorski\v{i} \cite{BL}, Buffington and Elliott \cite{BE2002}, Le and Wang \cite{LeWang},
%study American style options.

Despite the prominence of regime-switching models, the literature on Asian options within this context is scarce. %, and more generally exotic options
Some work has been done for the class of fixed-strike Asian options, see for instance Boyle and Draviam \cite{BD2007} and Dan et al. \cite{Danetal2016}.
%propose a numerical approach to pricing European and exotic instruments such as fixed-strike Asian and lookback options.
The pricing methods typically require solving a system of coupled PDEs.
%Recently, Chan and Zhu \cite{LS2014} derived expressions for floating-strike Asian put prices, for which the payoff is $\left(A_T/(T-t_0)-X_T\right)^+$,
%in (\ref{eq:AsianPutprice}) %of Asian put options with floating-strike
%when $t_0=0$ and the number of regimes is $m=2$.
%Chan and Zhu's approach involves solving a system of partial differential equations (PDEs), arising from the Feynman-Kac formula applied to the pricing function,
%and the so-called {\it homotopy analysis method}.
%Roughly, they construct auxiliary functions parameterized by $p\in[0,1]$ that solve a system of PDEs similar to the original one.
%The auxiliary functions coincide with the searched ones when $p=1$ and are written as a Taylor expansion around $p=0$.
%Thus, in order to compute their value, one needs to find the partial derivatives of all orders of the auxiliary functions.
%It turns out that each derivative solves a system of PDEs that has a closed-form solution.
In this paper, we explore an alternative approach to pricing floating-strike Asian call options with regime-switching,
%and hence put options as well thanks to the put-call parity (see Lemma \ref{eq:parity} below),
based on the fixed-point theorem for Banach spaces with the supremum norm and when the number of states $m$ is arbitrary.
The initial value in the algorithm is precisely the price of a fixed-strike Asian option without regime-switching,
which has been more extensively studied in the literature.
%Pricing methods for fixed-strike options without regime-switching are well-developed, including closed-form formulas.
For instance, Geman and Yor \cite{GY1993} were able to give an expression for the Laplace transform of a \textit{normalized} fixed-strike Asian call option by exploiting probabilistic properties of Bessel processes. The normalized Asian option involves the expectation of a function of Yor's process $A^\nu_t$, see (\ref{Yorprocs}) below.
Then the price of the option can be obtained by inversion of the Laplace transform, although they noted that such inversion was not easy.
Later on, Carr and Schr\"{o}der \cite{CS2004} built on Laplace transform techniques and provided an explicit integral representation of the price.
The same year, Linetsky \cite{Lin2004} took a different approach and showed that the normalized price
is the limit of up-and-out options on the diffusion $X$, % (Proposition 2 in \cite{Lin2004}),
each of which is given as a series representation of known special functions.
More recently, Cai et al. \cite{Caietal} obtained an algorithm to price Asian options based on an approximating continuous-time Markov chain sequence that converges to the underlying asset price process.
Other authors have provided price bounds, see for instance Rogers and Shi \cite{RogersShi} who use iterated conditional expectation.

The paper is organized as follows.
In Section \ref{Prelim}, convenient upper bounds of the call options are derived as well as
a symmetry relationship, in the context of no regime switching, between the starting floating-strike call and a fixed-strike put.
These upper bounds are crucial in order to ensure the relevant functions belong to the domain of a certain contraction operator after scaling.
We summarize the main results of the paper at the end of this section.
%The symmetry gives the initial point for a converging sequence that we construct in the successive approximations method.

Next, in Section \ref{Conditional} we split the functions $C$ and $C_K$ into two parts, one that restricts the payoff to the event that the Markov chain jumps before maturity and the other to the complementary event where the Markov chain does not jump in the lifetime of the option.
In Section \ref{Dist} we find the joint density of the pair $(Z_t,A_t)$, where $Z_t=\log(X_t/x)$,
given the information up to time $s$ for $t>s$ and given that the first jump time of the Markov chain after $s$ happens at time $t$.
We use this density in Section \ref{Aprox}, where we characterize the functions $C$ and $C_K$ as the limit (in the supremum norm) of a sequence whose initial point is in terms of the price of an Asian call option without regime-switching. The contraction operator is a nonlinear operator expressed as a triple integral that accounts for the jumps to different states before maturity.
%This gives insight into the structure of (\ref{eq:Asianprice}) as the sum of the price of an Asian option without regime-switching (as if the initial state was absorbing) and a function expressed as a triple integral that accounts for the jumps to different states before maturity.

The ideas in this paper are motivated by the method used by Yao et el. \cite{YZZ2006} applied to price vanilla European options.
The difficulty in our context stems from the fact that Asian options are path-dependent and the joint density of geometric Brownian motion and its integrated process is required.
Nonetheless, the fixed-point theorem approach works well in this setup and we are able to show that the rate of convergence of the sequence is geometric.
%Conditional on the Markov chain starting at $Y_0=i$,
Proof of preliminary lemmas appear in the Appendix.

\section{Preliminaries and main results} \label{Prelim}

It is known that European call options are bounded above by the current price of the underlying process.
This is also true for the floating-strike option in (\ref{eq:Asianprice}) and for the fixed-strike option in (\ref{eq:Asianprice_fixed}) up to a constant,
and will be used in the fixed-point approximation.
The proofs of the next lemmas are presented in the Appendix.

\begin{lemma} \label{lemma:bound}  %\cite{HendW2002} for a proof, I'll prove it here anyway.
For any initial condition $(s,x,a,i)$ with $s\in[0,T]$, $x>0$, $a\geq 0$,
\begin{equation} \label{eq:bound}
C(s,x,a,i)\leq x.
\end{equation}
\end{lemma}

Define the call option conditional on the chain having no jump in the interval $[s,T]$,
\[
C^0(s,x,a,i):=\mathbb{E}_{s,x,a,i}\left[ e^{-r(i)(T-s)} \left(X_T-\frac{A_T}{T-t_0}\right)^+ \,\mid\, Y_t=i,\,\forall t\in[s,T]\right].
\]

When the option is starting or forward-starting, it is possible to establish a symmetry between the associated floating-strike call option $C^0(s,x,0,i)$
and a fixed-strike Asian put option, for each $i\in \mathcal{M}$.
When the option is in-progress, it is equivalent to a \textit{generalized starting option} (see (\ref{eq:gralsymmetryC0}) below).
This type of symmetry results were studied, for instance, by Henderson and Wojakowski \cite{HendW2002} and Henderson et al. \cite{Hendetal2007} in the classical setup without regime switching.

The proof of the next lemma is included for completeness of presentation but similar arguments are used in \cite{Hendetal2007}.

\begin{lemma} \label{lemma:sym} %MODIFIED
For any initial condition $(s,x,a,i)$ with $s\in[0,T]$, $x>0$, $a\geq 0$,
\begin{equation} \label{eq:gralsymmetryC0}
C^0(s,x,a,i)=\mathbb{E}^*_{s,x,a,i}\left[e^{-\delta(T-s)} \left(x-\lambda\, X_T^*-\beta\,\frac{1}{T-s}\int_s^T X_u^*du \right)^+\right]
\end{equation}
where $\lambda=\frac{a}{x(T-t_0)}$ and $\beta=\frac{T-s}{T-t_0}$ and the expectation $\mathbb{E}^*$ is with respect to an equivalent martingale measure $P^*$  %defined in the proof of Lemma \ref{lemma:bound},
under which
%\[
%x A_T^*=\int_{t_0}^T X^*_s\,ds,
%\]
%and
$X^*$ solves the stochastic differential equation
\[
dX^*_t=X^*_t[(\delta-r(i))dt+\sigma(i)dB^*_t], \quad X^*_{s}=x, \qquad t\geq s.
\]
In particular, if the option is starting ($s=t_0$) or forward-starting ($s<t_0$)
then $a=0$ and the floating-strike call option $C^0$ is equivalent to a {\it fixed-strike} put option. Specifically,
\begin{equation}\label{eq:symmetryC0}
C^0(s,x,0,i)=\mathbb{E}^*_{s,x,0,i}\left[e^{-\delta(T-t_0)} \left(x-\beta \frac{A_T^*}{T-t_0} \right)^+\right]
\end{equation}
where $A_T^*=\int_{t_0}^T X_u^*\,du$, and $\beta\geq 1$.
\end{lemma}

There are well-known methods for fixed-strike options without switching coefficients as in (\ref{eq:symmetryC0}), some works have been cited in the introduction.
Using any of such methods, in conjunction with the so-called put-call parity for fixed-strike Asian options (see \cite[p.220]{Kwok}),
the value of $C^0$ in (\ref{eq:symmetryC0}) can be computed.
In contrast, methods to compute (\ref{eq:gralsymmetryC0}) explicitly are less accessible.
In a recent work by Funahashi and Kijima \cite{FK2017}, they provide an approximation method for generalized Asian options by applying a so-called chaos expansion approach.
Monte Carlo methods can be used as a benchmark when there are no closed-form formulas, see \cite{Lap&Tem}.
%Henderson et al. \cite{Hendetal2007} provided an upper bound in terms of a vanilla European option and a fixed-strike Asian option.

Now we turn to the case of fixed-strike options.

\begin{lemma} \label{lemma:bound_fixed}  %\cite{HendW2002} for a proof, I'll prove it here anyway.
Fix an initial condition $(s,x,a,i)$ with $s\in[0,T]$, $x>0$, $a\geq 0$.
\begin{itemize}
    \item[(i)] If $s\leq t_0$ then $a=0$ and $C_K(s,x,0,i)\leq x$.

    \item[(ii)] If $s>t_0$ then
    \[
    C_K(s,x,a,i)\leq \frac{a}{T-t_0}+x.
    \]
\end{itemize}
\end{lemma}

We conclude this section with a summary of the main results in the paper, without stating the technical details which we examine in the subsequent sections.
Consider the Banach space $\mathcal{S}$ of all bounded measurable functions $H:E\mapsto \mathbb{R}$, $E=[0,T]\times \mathbb{R}\times \mathbb{R}_+\times \mathcal{M}$, with the supremum norm
\[
||H||:=\sup_{(s,z,a,i)\in E}|H(s,z,a,i)|.
\]
Let $F: \mathcal{S} \mapsto \mathcal{S}$ be defined by
\begin{equation} \label{eq:mapping}
\begin{split}
F(H)(s,z,a,i)& := \\
& \hspace{-2cm} \sum_{j\neq i}q_{i j} \int_s^T e^{-[q_i+r(i)](t-s)}\, \int_a^\infty \int_{-\infty}^\infty e^{z'}H(t,z+z',a',j)\,\psi(z',a')\, dz'\,da'\,dt
\end{split}
\end{equation}
\noindent where $\psi$ is a joint density function to be derived later (see Proposition \ref{prop:joint_density} below).

We now state the main result.

\begin{theorem} [Contraction] \label{thm:main_intro}
\begin{itemize}
    \item[(i)] $F$ is a contraction mapping on $\mathcal{S}$.

    \item[(ii)] If $H_0,H\in\mathcal{S}$ and $H$ solves the equation
    \[
    H(s,z,a,i)=F(H)(s,z,a,i)+H_0(s,z,a,i)
    \]
    then $H$ is the only solution.

    \item[(iii)] The sequence $\{H_n\}_{n=0}^\infty$, with
    \[
    H_{n+1}(s,z,a,i)=F(H_n)(s,z,a,i)+H_0(s,z,a,i),
    \]
    converges to the fixed-point $H$ with geometric rate of converge
    \begin{equation}\label{eq:rate_conv}
    \rho=\max_{i\in \mathcal{M}} \sum_{j\neq i}\frac{q_{i j}}{q_i+\delta} \left(1-e^{-(q_i+\delta)(T-s)}\right)<1.
    \end{equation}
\end{itemize}
\end{theorem}

We specialize to the floating-strike Asian options below.

\begin{theorem}[Floating-strike option as a fixed-point]\label{thm:floating}
Define the functions $g,g_0 \in \mathcal{S}$ by
\[%begin{equation} \label{eq:seq}
\begin{aligned}
g(s,z,a,i)  &:=e^{-z}C(s,e^z,a,i), \\
%g^0(s,z,a,i) :=e^{-z}C^0(s,e^z,a,i).
g_0(s,z,a,i)&:=e^{-q_i(T-s)}e^{-z}C^0(s,e^z,a,i).
\end{aligned}
\]%end{equation}
%By Lemma \ref{lemma:bound}, we know that both $g$ and $g^0$ are bounded above by $1$.
Then $g$ is the fixed-point of
    \[
    g(s,z,a,i)=F(g)(s,z,a,i)+g_0(s,z,a,i).
    \]
Moreover, the sequence $\{g_n\}_{n=0}^\infty$, with $g_{n+1}=F(g_n)+g_0$
converges to $g$ with geometric rate of converge $\rho$ in (\ref{eq:rate_conv}).
\end{theorem}

For fixed-strike Asian options, the initial condition of the approximating sequence depends on whether the option is starting, forward-starting or in-progress.
More precisely, we have the following statement.

\pagebreak
\begin{theorem}[Fixed-strike option as a fixed-point]\label{thm:fixed}

\noindent (i) If $s\leq t_0$, define the functions $h,h_0 \in \mathcal{S}$ by
    \[
    \begin{aligned}
    h(s,z,a,i)  &:=e^{-z}C_K(s,e^z,a,i), \\
    h_0(s,z,a,i) &:=e^{-q_i(T-s)}e^{-z}C_K^0(s,e^z,a,i).
    \end{aligned}
    \]
    Then $h$ is the fixed-point of
    \[%begin{equation}\label{eq:fixed_starting}
    h(s,z,a,i)=F(h)(s,z,a,i)+h_0(s,z,a,i).
    \]%end{equation}
    Moreover, the sequence $\{h_n\}_{n=0}^\infty$, with $h_{n+1}=F(h_n)+h_0$
    converges to $h$ with geometric rate of converge $\rho$ in (\ref{eq:rate_conv}).

\noindent (ii) If $s>t_0$, define the functions $\tilde{h},\tilde{h}_0 \in \mathcal{S}$ by
    \[
    \begin{split}
    \tilde{h}(s,z,a,i)  &:=e^{-z}\left(C_K(s,e^z,a,i)-\frac{a}{T-t_0}\right), \\
    \tilde{h}_0(s,z,a,i) &:=e^{-q_i(T-s)}e^{-z}\left(C^0_K(s,e^z,a,i)-\frac{a}{T-t_0}\right)+\tilde{h}^1(s,z,a,i)
    \end{split}
    \]
    where
    {\small
    \begin{equation} \label{extra_term}
    \tilde{h}^1(s,z,a,i)=\frac{ae^{-z}}{T-t_0}\left[\frac{q_i}{q_i+r(i)}\left(1-e^{(q_i+r(i))(T-s)}\right)+e^{-q_i(T-s)}-1\right].
    \end{equation}}
    Then $\tilde{h}$ is the fixed-point of
    \[%begin{equation}\label{eq:fixed_starting}
    \tilde{h}(s,z,a,i)=F(\tilde{h})(s,z,a,i)+\tilde{h}_0(s,z,a,i).
    \]
    Moreover, the sequence $\{\tilde{h}_n\}_{n=0}^\infty$, with $\tilde{h}_{n+1}=F(\tilde{h}_n)+\tilde{h}_0$
    converges to $\tilde{h}$ with geometric rate of converge $\rho$ in (\ref{eq:rate_conv}).
\end{theorem}

\section{Conditioning on the first jump time} \label{Conditional}

Let us fix the current time $s\in[0,T]$ throughout the rest of the paper.
Conditional on $Y_s=i$, let $\tau$ denote the first jump time of the Markov chain $Y$ after time $s$, that is
\[
\tau=\inf\{t>s:\, Y_t\neq i\}.
\]
We know that $\tau$ has exponential distribution with parameter $q_i$.
Plainly,
\begin{equation} \label{eq:bondpriceSplit}
\begin{aligned}
C(s,x,a,i) & = e^{-q_i(T-s)}C^0(s,x,a,i) \\
        & + \mathbb{E}_{s,x,a,i}\left[\, e^{-\int_s^T r(Y_u) du}\left(X_T-\frac{A_T}{T-t_0}\right)^+\,\mathbbm{1}(\tau\leq T)\,\right].
\end{aligned}
\end{equation}
Likewise,
\begin{equation} \label{eq:bondpriceSplit_fixed}
\begin{aligned}
C_K(s,x,a,i) & = e^{-q_i(T-s)}C_K^0(s,x,a,i) \\
        & + \mathbb{E}_{s,x,a,i}\left[\, e^{-\int_s^T r(Y_u) du}\left(\frac{A_T}{T-t_0}-K\right)^+\,\mathbbm{1}(\tau\leq T)\,\right].
\end{aligned}
\end{equation}

Notice that by conditioning the expectation in (\ref{eq:bondpriceSplit}) on the jump time $\tau=t$, $t\geq s$, we can write it as
\[
\begin{aligned}
\mathbb{E}_{s,x,a,i} & \left[\, e^{-\int_s^T r(Y_u) du}\left(X_T-\frac{A_T}{T-t_0}\right)^+\,\mathbbm{1}(\tau\leq T)\,\right]  \\
& \hspace{-1cm} = \int_s^T q_i e^{-q_i(t-s)} \mathbb{E}_{s,x,a,i} \left[\, e^{-\int_s^T r(Y_u) du}\left(X_T-\frac{A_T}{T-t_0}\right)^+\,\mid \tau=t\right]  dt \\
%& \hspace{-1cm} = \int_s^T q_i e^{-q_i(t-s)} \mathbb{E}_{s,x,a,i}
%        \left[\, e^{-r(i)(t-s)}  \mathbb{E}_{s,x,a,i} \left[\, e^{-\int_t^T r(Y_u) du}\left(X_T-\frac{A_T}{T-t_0}\right)^+\,\mid X_t,A_t,Y_t\right]\mid \tau=t\right]  dt \\
& \hspace{-1cm} = \int_s^T q_i e^{-q_i(t-s)} \mathbb{E}_{s,x,a,i} \left[\, e^{-r(i)(t-s)} C(t,X_t,A_t,Y_t) \mid \tau=t\right]  dt
\end{aligned}
\]
by virtue of the Markov property of $(t,X_t,A_t,Y_t)$. A similar argument holds for the fixed-strike case.

In what follows it will be convenient to work with the process
\[
Z_t:=\int_s^t\sigma(Y_u)dB_u + \int_s^t \left(r(Y_u)-\delta-\frac{1}{2}\sigma^2(Y_u)\right)du, \qquad t\geq s
\]
so that
\begin{equation}
X_t=\exp(z+Z_t), \quad z:=\ln(x).
\end{equation}

\subsection{Density of $(Z_t,A_t)$} \label{Dist}% conditional on $\tau=t$
%We next develop a numerical approximation for the value of $C(s,x,a,i)$ by means of the distribution of $(X_s,Y_s)$.

Conditional on $X_s=e^z, A_s=a, Y_s=i$ and $\tau=t$, it follows that
\[
Z_t\stackrel{law}{=}\sigma(i)B_{t-s}+\nu(i)(t-s), \qquad \nu(i):=r(i)-\delta-\frac{1}{2}\sigma^2(i)
\]
and
\[
A_t=a+\int_s^t X_u du\stackrel{law}{=}a+e^z\int_0^{t-s} e^{\sigma(i)B_{u}+\nu(i)u} du.
\]
The pair $(Z_t,A_t)$ is independent of $Y_t$
and its distribution can be explicitly computed. To this end, define
\begin{equation} \label{Yorprocs}
A_t^{\nu}:=\int_0^t e^{2(B_u+\nu u)}du, \qquad \nu\in \mathbb{R}.
\end{equation}
The following preliminary result is due to Yor \cite{Yor1992}.

\begin{lemma}\label{lemma:Yor}
We have
\[
P(A_t^\nu \in dw \mid B_t+\nu t=z)=f(t,z,w)dw
\]
where
\[
\frac{1}{\sqrt{2 \pi t}}\exp\left( -\frac{z^2}{2t}\right)f(t,z,w)=\frac{1}{w}\exp\left(-\frac{1}{2w}(1+e^{2z})\right)\theta_{e^z/w}(t)
\]
and
\[
\theta_r(t)=\frac{r}{\sqrt{2\pi^3 t}}\exp\left(\frac{\pi^2}{2t}\right)\,\int_0^\infty e^{-y^2/2t}\,e^{-r\cosh(y)} \sinh(y)\,\sin(\pi y/t) \,dy.
\]
\end{lemma}
We refer to Proposition 2 in \cite{Yor1992} for a proof.

\begin{proposition}\label{prop:joint_density}
The joint density $\psi(z',a')$ of the pair $(Z_t,A_t)$,
conditional on $X_s=x=e^z, A_s=a, Y_s=i$, and $\tau=t$, %given $X_s=x,A_s=a,Y_s=i$,
is given by
\[
\psi(z',a') =
\frac{\sigma^2(i)}{4}e^{-z}\, f\left(t', z', w(a')\right)\, \phi\left(\frac{z'-2\nu t'}{2\sqrt{t'}}\right) 1_{\{\mathbb{R}\times [a,\infty]\}},
\]
%for $(x',a')\in [0,\infty]\times [a,\infty]$,
with $w(a')=\frac{\sigma^2(i)}{4}e^{-z}(a'-a)$ and $t'=\frac{\sigma^2(i)}{4}(t-s)$.

\end{proposition}
\begin{proof}
A direct consequence of Lemma \ref{lemma:Yor} is that
\[
P(\,2(B_t+\nu t)\in dz, A_t^\nu \in dw )=f\left(t,\frac{z}{2},w\right)\phi\left(\frac{z-2\nu t}{2\sqrt{t}}\right)dw\,dz
\]
where $\phi(\cdot)$ is the density of a standard normal distribution.

Henceforth, conditional on $X_s=e^z, A_s=a, Y_s=i$, and $\tau=t$,
and writing $P(\cdot \mid X_s=e^z, A_s=a, Y_s=i,\tau=t)=P(\cdot)$ for short,
we have
\[
\begin{aligned}
P(Z_t\leq z', \;A_t\leq a')& \\
& \hspace{-3cm}  =  P\left(\,\sigma(i)B_{t-s}+\nu(i)(t-s)\leq z', \; \int_0^{t-s} e^{\sigma(i)B_u+\nu(i)u}du  \leq e^{-z}(a'-a) \right) \\
& \hspace{-3cm}  =  P\left(\,2(B_{t'}+\nu t')\leq z', \; \int_0^{t'} e^{2(B_u+\nu u)} \,du  \leq \frac{\sigma^2(i)}{4}e^{-z}(a'-a)\right)
\end{aligned}
\]
where we used the scaling property $\sigma(i)B_{t-s}\stackrel{law}{=} B_{\sigma^2(i)(t-s)}$ and the change of variables
\[
t'\equiv \frac{\sigma^2(i)}{4}(t-s), \qquad \nu \equiv \frac{2\nu(i)}{\sigma^2(i)}.
\]
Finally,
\[
P(Z_t\leq z', \;A_t\leq a')=\int_{-\infty}^{z'} \int_0^{w(a')} f\left(t',\frac{z}{2},w\right)\phi\left(\frac{z-2\nu t'}{2\sqrt{t'}}\right)dw\,dz
\]
and a further change of variable from $w$ to $a'$ concludes the proof.
\end{proof}

%\section{Successive approximations for Asian call options} \label{Aprox}
\section{Fixed-point}\label{Aprox}

\subsection{The main contraction theorem}
The goal of this subsection is to show Theorem \ref{thm:main_intro}.

\begin{proposition}\label{prop:part1}
$F$ is a contraction mapping on $\mathcal{S}$.
\end{proposition}
\begin{proof}
For each $(s,z,a,i)$ and $t\geq s$ fixed, %, recall from the previous section the change of variable $t'=\sigma^2(i)(t-s)$ and that $a,z$ are implicit in the density $\psi$.
it follows that
\[ %Expectation of a exponential of a normal random variable with mean $\nu t'$ and variance $t'$
\int_a^\infty \int_{-\infty}^\infty e^{z'}\,\psi(z',a')\, dz'\,da'=e^{(r(i)-\delta)(t-s)},
\]
and so
\[
\begin{aligned}
\rho(i)  &:=\sum_{j\neq i}q_{i j} \int_s^T e^{-[q_i+r(i)](t-s)}\, \int_a^\infty \int_{-\infty}^\infty e^{z'}\,\psi(z',a')\, dz'\,da'\,dt \\
         &=\sum_{j\neq i}q_{i j}\int_s^T e^{-(q_i+\delta)(t-s)}dt
            = \sum_{j\neq i}\frac{q_{i j}}{q_i+\delta} \int_s^T (q_i+\delta) e^{-(q_i+\delta)(t-s)}dt \\
         & =\sum_{j\neq i}\frac{q_{i j}}{q_i+\delta} \left(1-e^{-(q_i+\delta)(T-s)}\right)
         %& <\sum_{j\neq i}\frac{q_{i j}}{q_i+\delta} \\
         %& =\frac{1}{q_i+\delta}\sum_{j\neq i}q_{i j} = \frac{q_i}{q_i+\delta}<1
         <1.
\end{aligned}
\]
Then,
\[
\rho:=\max_{i\in \mathcal{M}}\rho(i)<1
\]
which yields the inequality $||F(H)|| \leq \rho ||H||$, as desired.
\end{proof}

Parts (ii) and (iii) of Theorem \ref{thm:main_intro} are immediate from Proposition \ref{prop:part1}.

\begin{corollary} \label{prop:part2}
If $H_0,H\in\mathcal{S}$ and $H$ solves the equation
    \begin{equation} \label{eq:gAsianF}
    H(s,z,a,i)=F(H)(s,z,a,i)+H_0(s,z,a,i)
    \end{equation}
then $H$ is the only solution.
\end{corollary}
\begin{proof}
%Observe that equation (\ref{eq:gAsian}) may be rewritten as (\ref{eq:gAsianF}).
Since $F$ is a contraction so is the translation mapping $F(\cdot)+H_0$.
Henceforth, $F(\cdot)+H_0$ has a fixed point thanks to the Banach Fixed-Point Theorem.
This in turn implies the uniqueness.
\end{proof}

\begin{corollary} \label{cor:main}
 The sequence $\{H_n\}_{n=0}^\infty$, with
 \[
 H_{n+1}(s,z,a,i)=F(H_n)(s,z,a,i)+H_0(s,z,a,i),
 \]
 converges to the fixed-point $H$ with geometric rate of converge
    \[
    \rho=\max_{i\in \mathcal{M}} \sum_{j\neq i}\frac{q_{i j}}{q_i+\delta} \left(1-e^{-(q_i+\delta)(T-s)}\right)<1.
    \]
\end{corollary}
\begin{proof}
Thanks to Corollary \ref{prop:part2}, $\{H_n\}_{n=0}^\infty$ converges to $H$ in the supremum norm.
We have that
$H_{n+1}-H=F(H_n)-F(H)=F(H_n-H)$. Then using the fact that $F$ is a contraction,
\[
||H_{n+1}-H||\leq \rho ||H_{n}-H||
\]
and
\[
||H_{n+1}-H_n||\leq \rho^n ||H_1-H_0||
\]
where $\rho$ is defined in Proposition \ref{prop:part1}.
\end{proof}

\subsection{Floating-strike case} \label{sub:float}
In this subsection, we show Theorem \ref{thm:floating}.

Consider the functions
\[
g(s,z,a,i)  =e^{-z}C(s,e^z,a,i), \qquad
g^0(s,z,a,i) =e^{-z}C^0(s,e^z,a,i).
\]

%In what follows it will be convenient to work with the process
%\[
%Z_t:=\int_s^t\sigma(Y_u)dB_u + \int_s^t \left(r(Y_u)-\delta-\frac{1}{2}\sigma^2(Y_u)\right)du, \qquad t\geq s
%\]
%so that
%\begin{equation}
%X_t=\exp(z+Z_t), \quad z:=\ln(x).
%\end{equation}
Observe that (\ref{eq:bondpriceSplit}) can be written as
\begin{equation} \label{eq:gAsian}
\begin{aligned}
g(s,z,a,i)&=e^{-q_i(T-s)}g^0(s,z,a,i) \\
        &\hspace{-1.5cm} +\int_s^T q_ie^{-q_i(t-s)}\mathbb{E}_{s,x,a,i} \left[\, e^{-r(i)(t-s)} e^{Z_t} g(t,z+Z_t,A_t,Y_t) \mid \tau=t\right]  dt.
\end{aligned}
\end{equation}

Moreover,
\[
\begin{aligned}
\mathbb{E}_{s,x,a,i} \left[\, e^{-r(i)(t-s)} e^{Z_t} g(t,z+Z_t,A_t,Y_t) \mid \tau=t\right] & \\
& \hspace{-6cm} = \sum_{j\neq i} \frac{q_{i j}}{q_i} \int_a^\infty \int_{-\infty}^\infty e^{-r(i)(t-s)} e^{z'}g(t,z+z',a',j)\,\psi(z',a')\, dz'\,da'
\end{aligned}
\]
where $\psi$ is the density in Proposition \ref{prop:joint_density},
and the second term on the right-hand side of equation (\ref{eq:gAsian}) then reads
\[
\sum_{j\neq i}q_{i j} \int_s^T e^{-[q_i+r(i)](t-s)}\int_a^\infty \int_{-\infty}^\infty e^{z'}g(t,z+z',a',j)\,\psi(z',a')\, dz'\,da'\,dt.
\]
This is the mapping $F$ as defined in (\ref{eq:mapping}) and we can further write
\begin{equation} \label{eq:floating}
g(s,z,a,i)=F(g)(s,z,a,i)+e^{-q_i(T-s)}g^0(s,z,a,i).
\end{equation}

\subsection{Fixed-strike case}
In this subsection, we show Theorem \ref{thm:fixed}.

For $s\leq t_0$ (starting and forward-starting options), consider the functions
\[
h(s,z,a,i)  =e^{-z}C_K(s,e^z,a,i), \qquad
h^0(s,z,a,i) =e^{-z}C_K^0(s,e^z,a,i).
\]
In this case, similar in structure to the floating-strike, we obtain the equation,
\begin{equation}\label{eq:fixed_starting}
h(s,z,a,i)=F(h)(s,z,a,i)+e^{-q_i(T-s)}h^0(s,z,a,i), \qquad s\leq t_0.
\end{equation}

For $s>t_0$ (in-progress options), consider the functions
\[
\begin{split}
\tilde{h}(s,z,a,i)  &=e^{-z}\left(C_K(s,e^z,a,i)-\frac{a}{T-t_0}\right), \\
\tilde{h}^0(s,z,a,i) &=e^{-z}\left(C^0_K(s,e^z,a,i)-\frac{a}{T-t_0}\right),
\end{split}
\]
so that (\ref{eq:bondpriceSplit_fixed}) can be written as
{\small \[
\begin{aligned}
\tilde{h}(s,x,a,i) +\frac{ae^{-z}}{T-t_0}& = e^{-q_i(T-s)}\left(\tilde{h}^0(s,x,a,i)+\frac{ae^{-z}}{T-t_0}\right) \\
& \hspace{-3.5cm} + \int_s^T q_ie^{-q_i(t-s)}\mathbb{E}_{s,x,a,i}\left[e^{-r(i)(t-s)}\left(e^{Z_t}\tilde{h}(t,z+Z_t,A_t,Y_t)+\frac{ae^{-z}}{T-t_0}\right) \mid \tau=t\right]dt.
\end{aligned}
\]}
After some algebraic manipulation, we obtain the equation, for $s>t_0$,
\[%begin{equation}\label{eq:fixed_inprogr}
\tilde{h}(s,z,a,i)=F(\tilde{h})(s,z,a,i)+e^{-q_i(T-s)}\tilde{h}^0(s,z,a,i)+\tilde{h}^1(s,z,a,i),
\]%end{equation}
where the extra term is $\tilde{h}^1$ is in (\ref{extra_term}).

\subsection{Iteration}

Theorem \ref{thm:main_intro} provides an iterative method to approximate the Asian option functions.
For instance, we can approximate
\[
g(s,z,a,i)=e^{-z}C(s,e^z,a,i)
\]
with $z=\ln(x)$ by a fixed small error, say $\epsilon>0$:
\[
\begin{aligned}
g_0(s,z,a,i) & =e^{-(q_i(T-s)+z)}\,C^0(s,e^{z},a,i), \\
g_{n+1}(s,z,a,i) & = F(g_n)(s,z,a,i)+g_0(s,z,a,i), \qquad n\geq 0\\
\mbox{If} &\quad ||g_{n+1}-g_n||<\epsilon,\quad \mbox{stop}.
\end{aligned}
\]
Observe that the larger the dividend rate $\delta$, the faster the convergence. This can be implied from the expression for $\rho(i)$ above.

While the algorithm is theoretically appealing and provides an alternative to solving a certain system of PDEs as it is usual in option pricing,
we should mention that in order to approximate the function $g$ (and then $C)$, it is necessary to compute first the initial function $C^0$ for the iteration.
%This function corresponds to the price of either a generalized starting option or a fixed-strike Asian put option without regime-switching, see Lemma \ref{lemma:sym}.
A good estimation of $C^0$ is important to avoid amplifying the error in the iteration. In this regards, it is known
that the so-called Hartman-Watson density appearing in the definition of the density $\psi$ is indeed difficult to implement. %For the former, Lemma \ref{lemma:Upperbound} provides an upper bound.
To analyze the effect of the error incurred by such approximation, suppose that the initial function for the iteration is, say $\tilde{g}_0 \in \mathcal{S}$.
Then the mapping $F(\cdot)+\tilde{g}_0$ is also a contraction with the same rate of convergence $\rho$.
Moreover, the fixed-point theorem implies that the sequence $\{\tilde{g}_n\}_{n\geq 0}$ defined by
\[
\tilde{g}_{n+1}:=F(\tilde{g}_n)+\tilde{g}_0, \qquad n\geq 0
\]
converges to a fixed-point, say $\tilde{g}$, which solves the equation
\[
\tilde{g}=F(\tilde{g})+\tilde{g}_0.
\]
Hence, we can check that
\[
||\tilde{g}-g||\leq \frac{1}{1-\rho}||\tilde{g}_0-g_0||.
\]
In other words, the accuracy of the algorithm depends proportionally on the error incurred at the initial step.
%Specifically, we look at the lower and upper bounds of $C^0$ and study the effect of the operator $F$ in the iteration.

\appendix

\section{Proofs}

\begin{proof}[Proof of Lemma \ref{lemma:bound}]
Define the probability measure $P^*$ equivalent to $P$ via the Radon Nikodym derivative
\[
\frac{dP^*}{dP}\mid_{\mathcal{F}_T}= \mathcal{E}_T
\]
where
\[
\mathcal{E}_t:=\exp\left( \int_0^t \sigma(Y_u)dB_u-\frac{1}{2}\int_0^t \sigma^2(Y_u)du \right).
\]
%and set $A^*_T:=A_T/X_T$.

The call option satisfies
\[
\begin{aligned}
\frac{C(s,x,a,i)}{x}    & =\mathbb{E}_{s,x,a,i}\left[ \frac{e^{-\int_s^T r(Y_u)du} X_T}{x} \left(1-\frac{1}{T-t_0}\frac{A_T}{X_T}\right)^+ \right] \\
                        & =\mathbb{E}_{s,x,a,i}\left[ \mathcal{E}_T \mathcal{E}^{-1}_s \,e^{-\delta(T-s)} \left(1-\frac{1}{T-t_0}\frac{A_T}{X_T}\right)^+ \right] \\
                        & =\mathbb{E}_{s,x,a,i}^*\left[e^{-\delta(T-s)}\left(1-\frac{1}{T-t_0}\frac{A_T}{X_T}\right)^+ \right]\leq 1.
\end{aligned}
\]
where the expectation $\mathbb{E}^*$ is with respect to $P^*$. % and the second equality follows by the abstract Bayes formula (see \cite{MR2005}).
The result is now clear.
\end{proof}

\begin{proof}[Proof of Lemma \ref{lemma:sym}]
Following up the proof of Lemma \ref{lemma:bound}, we have that
\[
C^0(s,x,a,i)=\mathbb{E}_{s,x,a,i}^*\left[e^{-\delta(T-s)}\left(x-\frac{x}{T-t_0}\frac{A_T}{X_T}\right)^+\,\mid\, Y_t=i, \forall t\in[s,T] \right],
\]
and $\hat{B}_u=B_u-\int_0^u\sigma(Y_s)ds$ is a Brownian motion under $P^*$. Here,
\[
x\frac{A_T}{X_T} =\frac{x}{X_T}\left(a+\int_s^T X_u\,du\right).
\]

The process $(B^*_u)_{s\leq u\leq T}$, defined by $B^*_u:=B^*_{s}+\hat{B}_{T+s-u}-\hat{B}_T$ with $B^*_{s}$ a constant,
is also a Brownian motion under $P^*$ starting at $B^*_{s}$.
Now, conditional on $X_s=x$, $A_s=a$, and $Y_t=i$ for all $t\in[s,T]$,
\[
\begin{aligned}
\frac{x}{X_T}   & =\exp\left(\sigma(i)(\hat{B}_s-\hat{B}_T)+\left[r(i)-\delta+\frac{\sigma^2(i)}{2}\right](s-T)\right) \\
                & \stackrel{law}{=}\exp\left(\sigma(i)(B^*_T-B^*_s)+\left[\delta-r(i)-\frac{\sigma^2(i)}{2}\right](T-s)\right)
\end{aligned}
\]
and
\[
\begin{aligned}
x\int_s^T \frac{X_u}{X_T} & =\int_{s}^T x\exp\left(\sigma(i)(\hat{B}_u-\hat{B}_T)+\left[r(i)-\delta+\frac{\sigma^2(i)}{2}\right](u-T) \right)du \\
                & \hspace{-1cm} \stackrel{law}{=}\int_{s}^T x\exp\left(\sigma(i)(B^*_{T+s-u}-B^*_{s})+\left[\delta-r(i)-\frac{1}{2}\sigma^2(i)\right](T-u) \right)du \\
                & \hspace{-1cm} =\int_{s}^T x\exp\left(\sigma(i)(B^*_w-B^*_{s})+\left[\delta-r(i)-\frac{1}{2}\sigma^2(i)\right](w-s) \right)dw
                %& =\int_{s}^T X^*_w dw
\end{aligned}
\]
where %we used that $Y_t=i$ for all $t\in[s,T]$, and
the third equality is obtained after the change of variable $w=T+s-u$. Therefore, $C^0(s,x,a,i)$ is given by
\[
\mathbb{E}^*_{s,x,a,i}\left[e^{-\delta(T-s)} \left(x-\frac{a}{x(T-t_0)}X_T^*-\frac{T-s}{T-t_0}\frac{1}{T-s}\int_s^T X_u^*du \right)^+\right]
\]
where the underlying process $X^*$ follows
\[
dX^*_t=X^*_t[(\delta-r(i))dt+\sigma(i)dB^*_t], \quad X^*_{s}=x, \qquad t\geq s.
\]
Defining the parameters $\lambda=\frac{a}{x(T-t_0)}$ and $\beta=\frac{T-s}{T-t_0}$ the proof is complete.
\end{proof}

\begin{proof}[Proof of Lemma \ref{lemma:bound_fixed}]
Part (i). Let $s\leq t_0$. Then
\[
\begin{split}
C_K(s,x,0,i)    & \leq \mathbb{E}_{s,x,0,i}\left[e^{-\int_s^T r(Y_u)du}\frac{A_T}{T-t_0}\right] \\
                %& =\frac{1}{T-t_0}\int_{t_0}^T \mathbb{E}_{s,x,0,i}\left[e^{-\int_s^T r(Y_u)du}X_t\right]dt \\
                & \leq \frac{1}{T-t_0}\int_{t_0}^T \mathbb{E}_{s,x,0,i}\left[e^{-\int_s^t r(Y_u)du}X_t\right]dt \\
                & \leq \frac{x}{T-t_0}\int_{t_0}^Te^{-\delta(t-s)}dt \leq x.
\end{split}
\]
Part(ii). Let $s>t_0$. Then
\[
\begin{split}
C_K(s,x,a,i)    & \leq \mathbb{E}_{s,x,a,i}\left[e^{-\int_s^T r(Y_u)du}\left(\frac{a+\int_s^T X_t dt}{T-t_0}\right)\right] \\
                &\leq \frac{a}{T-t_0}
                +\left(\frac{T-s}{T-t_0}\right)\frac{1}{T-s}\int_{s}^T \mathbb{E}_{s,x,a,i}\left[e^{-\int_s^t r(Y_u)du}X_t\right]dt \\
                & \leq \frac{a}{T-t_0}+\left(\frac{T-s}{T-t_0}\right)x
                 \leq \frac{a}{T-t_0}+x.
\end{split}
\]
\end{proof}
% ------------------------------------------------------------------------

\subsection*{Acknowledgment}
I thank the referees for their valuable comments, it has significantly improved this work.

% ------------------------------------------------------------------------
\end{document}